\documentclass[journal, 10pt, twoside]{IEEEtran}

\usepackage{times}
\usepackage{ifthen}
\usepackage[english]{babel}
\usepackage{amsmath, amsfonts, bm}
\usepackage{graphicx}
\usepackage{color}
\usepackage[utf8]{inputenc}
\usepackage{cite}
\usepackage{dblfloatfix}
\usepackage{units}  

\usepackage{amsthm}

\newtheorem{theorem}{Theorem}
\newtheorem{cor}{Corollary}
\theoremstyle{definition}

\newcommand{\ma}[1]{\bm{ #1 }} 
\newcommand{\compl}{\mathbb{C}}
\newcommand{\real}{\mathbb{R}}
\newcommand{\comment}[1]{}

\newcommand{\realof}[1]{{\rm Re}\left\{ #1 \right\}}

\newcommand{\traceof}[1]{{\rm trace}\left\{ #1 \right\}}

\newcommand{\fronorm}[1]{\left\|#1\right\|_{\rm F}}

\newcommand{\trans}{{\rm T}}

\newcommand{\herm}{{\rm H}}

\newcommand{\expvof}[1]{{\mathbb{E}}\hspace{0mm}\left\{ #1 \right\}}
\newcommand{\expvofsub}[2]{{\mathbb{E}}_{#1}\hspace{0mm}\left\{ #2 \right\}}

\IEEEoverridecommandlockouts

\title{Design and Analysis of Compressive Antenna Arrays for Direction of Arrival Estimation}
\author{Mohamed Ibrahim$^*$, 
		    Venkatesh Ramireddy,
		    Anastasia  Lavrenko, 
		    Jonas K\"onig, 
		    Florian R\"omer,
		    Markus Landmann,
		    Marcus Grossmann, 
		    Giovanni Del Galdo,
		    and Reiner S. Thom\"a
\thanks{
The authors 
M.~Ibrahim, V.~Ramireddy,  A.~Lavrenko, J.~K\"onig, F.~R\"omer, G.~Del~Galdo and R.~S.~Thom\"a are with 
Ilmenau University of Technology,
P.~O.~Box 10~05~65, 98684 Ilmenau, Germany.
The authors G.~Del~Galdo, R.~S.~Thom\"a,
M.~Landmann, and M.~Grossmann are with
Fraunhofer Institute for Integrated Circuits IIS,
Helmholtzplatz~2, 98683 Ilmenau, Germany. Parts of this work have been presented as conference papers at the 40th International Conference on Acoustics, Speech and Signal Processing (ICASSP),  Brisbane, Australia in April 2015 (Theorem 1 and Corollary 1) and at the 23rd European Signal Processing Conference (EUSIPCO), Nice, France in September 2015 (adaptive design approach from Section \ref{sec:adptv_dsn}).}
\thanks{$*$ corresponding author}
}

\usepackage[pdfpagelabels,breaklinks]{hyperref} 
\hypersetup{
   pdftitle={PUT TITLE HERE},
   pdfauthor={PUT AUTHORS HERE},
   pdfsubject={},
   pdfkeywords={},
   pdfview=FitH,
   pdfstartview=FitH,
   bookmarksnumbered
} 

\begin{document}

\maketitle

\noindent
{\bf {\bf \slshape Abstract \symbol{124}} 
In this paper we investigate the design of compressive antenna arrays for direction of arrival (DOA) estimation that aim to provide a larger aperture with a reduced hardware complexity by a linear combination of the antenna outputs to a lower number of receiver channels. We present a basic receiver architecture of such a compressive array and introduce a generic system  model that includes different options for the hardware implementation. We then discuss the design of the analog combining network that performs the receiver channel reduction, and propose two design approaches. The first approach is based on the spatial
correlation function which is a low-complexity scheme that in certain cases admits
a closed-form solution. The second approach is based
on minimizing the Cram\'{e}r-Rao Bound (CRB) with the constraint to limit the probability of false detection of paths to 
a pre-specified level. Our numerical simulations demonstrate the superiority of the proposed optimized compressive arrays compared to the sparse arrays of the same complexity and to compressive arrays with randomly chosen combining kernels. 
}
\medskip


\IEEEpeerreviewmaketitle

\section{Introduction} \label{sec_intro}



 Direction of arrival (DOA) estimation has been an active field of research for many decades \cite{KV:96}. In general, DOA estimation addresses the problem of locating sources which are radiating energy that is received by an array of sensors with known spatial positions \cite{HLVT:02}. Estimated DOAs are used in various applications like localization of transmitting sources, for direction finding \cite{VCK:95,CYH:02}, massive MIMO and 5G Networks \cite{CYH:02}, channel sounding and modeling \cite{RHST:00,TLR:04}, tracking and surveillance in radar \cite{BP:01}, and many others. 
A major goal in research on DOA estimation is to develop approaches that allow to minimize hardware complexity in terms of receiver costs and power consumption, while providing a desired level of estimation accuracy and robustness in the presence of multiple sources and/or multiple paths. Furthermore, the developed methods shall be applicable in practical applications with realistic antenna arrays whose characteristics often significantly vary from commonly considered ideal models \cite{LT:09}.

 In the last few decades, research on direction of arrival (DOA) estimation using array processing has largely focused on uniform arrays (e.g., linear and circular) \cite{HLVT:02} for which many efficient parameter estimation algorithms have been developed. Some well-known examples are ESPRIT \cite{RK:89}, MUSIC \cite{S:86} and Maximum Likelihood (ML)-based methods \cite{SOVM:96,TLR:04}. Note that ML-based methods are particularly suitable for realistic, non-ideal antenna arrays since they can easily account for the full set of parameters of the antenna array (e.g., antenna polarization, non-ideal antennas and array geometries, etc.). However, to perform well, the algorithms require to fulfill certain conditions on the sampling of the wavefront of the incident waves in the spatial domain.
Namely, the distance between adjacent sensors should be less than or equal to half a wavelength of the impinging planar wavefronts, otherwise it leads to grating lobes (sidelobes) in the spatial correlation function which correspond to  near ambiguities in the array manifold. At the same time,  to achieve DOA estimation with a high resolution, the receiving arrays should have a relatively large aperture \cite{HLVT:02}. This implies that arrays with a large number of antennas are needed to obtain a high resolution, which is not always feasible.


This limitation has triggered the development of arrays with inter-element spacing larger than half the impinging wave's wavelength combined with specific constraints to control the ambiguity problem in DOA estimation. Such arrays are usually called sparse arrays. In \cite{L:64}, it was proposed to constitute a non-uniform sparse array with elements spaced at random positions. However, using such random arrays will often result in an
unpredictable behavior of the sidelobes in the array's spatial correlation function. As a result, it is necessary to optimize the positions of the antenna elements in order to achieve a desired performance. An early approach towards that goal was the Minimum Redundancy Linear Array (MRLA) \cite{M:68}, where it is proposed to place the antenna elements such that the number of pairs of antennas that have the same spatial correlation properties are as small as possible. However, it is very difficult to construct an MRLA when the number of elements is relatively large \cite{MHC:09}. Some non-linear optimization methods like genetic algorithms \cite{H:94} and simulated annealing \cite{TM:99} have been regularly used to find optimum configurations for these sparse arrays. In more recent works, V-shaped arrays \cite{GA:09}, Co-Prime arrays \cite{VP11}, and Nested arrays \cite{PV10} have been proposed to extend the effective array aperture. 

Recently, compressed sensing (CS) \cite{CT:06,CRT:06,D:06} has been widely suggested for applications that exhibit sparsity in time, frequency or space to reduce the sampling efforts. The application of sparse recovery to DOA estimation has been considered for applications like localization of transmitting sources \cite{CGMC:08}, channel modeling \cite{FVT:09}, tracking and surveillance in radar \cite{E:10}, and many others.
It is highlighted in \cite{MCW:05} that if the electromagnetic field is modeled as a superposition of a few plane waves, the DOA estimation problem can be expressed as a sparse recovery problem. The main focus is to use the sparse recovery algorithms that became popular
in the CS field for the DOA estimation problem as an alternative
to existing parameter estimation algorithms~\cite{GCM:12,G:10,SBL:11,MZ:06}. 

Compressed sensing has also been suggested to be applied in the spatial domain (e.g., array processing and radar) with the main goal to reduce the complexity of the measurement process by using fewer RF chains and storing less measured data without the loss of any significant information. Hence, the idea of sparse random arrays with increased aperture size has been recently revisited and proposed to perform spatial compressed sensing \cite{HWKH:10,RHE:12,SAL:12,HL:14}. 


An alternative approach that attempts to apply CS to the acquisition of the RF signals that are used for DOA estimation has recently been proposed in \cite{WLP:09,GZS:11}. In particular, the CS paradigm can be applied in the spatial domain by employing $N$ antenna elements that are combined using an analog combining network to obtain a smaller number of $M < N$ receiver channels.
Since only $M$ channels need to be sampled and digitized, the hardware complexity remains comparably low while a larger aperture is covered which yields a better selectivity than a traditional, Nyquist ($\lambda/2$) spaced $L$-channel antenna array. 
In baseband, the operation of the combining network can be described by complex weights applied to the antenna outputs with a subsequent combination of the received signals from the antennas. 
The combining (measurement) matrix that contains the complex weights and the antenna array form an effective ``compressive" array whose properties define the DOA estimation performance.
In the field of ``CS-DOA'' it is usually advocated to draw the coefficients of the measurement matrix from a random distribution (e.g., Gaussian, Bernoulli) \cite{WLP:09,GZS:11}. Random matrices have certain guarantees for signal recovery in the noise-free case and provide some stability guarantees in the noisy case \cite{BDDW:08, TT:11, EK:12}. However, since no criterion is used to design them, it is likely that they provide sub-optimal performance \cite{IRD:15}. 

In this paper, we discuss the design and the performance of \textit{compressive} arrays employing linear combinations in the analog domain by means of a network of power splitters, phase shifters, and power combiners.
We present a basic receiver architecture of such a compressive array and introduce a generic system model that includes different
options for the hardware implementation. Importantly, the model reflects the implications for the noise sources. 
Particularly, a well-known source of the receiver noise is the low noise amplifier (LNA) that is usually placed at the antenna outputs to account for the power losses of the following distribution/combining network. Depending on the frequency range, the components of the analog combining network  (power combiners, power splitters, phase shifters) will induce additional losses which also have to be compensated by the LNAs. 
To name an example, some typical commercially available phase shifters for phased array radar applications can induce insertion losses between \unit[5 to 10]{dB} depending on the frequency range \cite{TQS:16:Online}. This motivates the need for the signal amplification prior the combining network.

Based on the generic system model we then discuss the design of the combining matrix, with the goal to obtain an array that is suitable for DOA estimation (i.e., minimum variance of DoA estimates and robustness in terms of low side lobe levels or low probability of false detections). We consider two design approaches. The first approach is based on the spatial correlation function which is a low-complexity scheme that in certain cases even admits a closed-form solution. The second approach is based on the minimization of the Cramér-Rao Bound (CRB).
CRB-minimizing array designs tend to result in high sidelobes in the spatial correlation function which may lead to false estimates. 
In order to be able to constrain this effect, we analytically derive the probability to detect a false peak (sidelobe) for a given array manifold.
We then use this expression as an additional constraint in our design, thus limiting the probability of false detection to a pre-specified level.
Our numerical simulations demonstrate that both proposed design approaches have a significant performance improvement compared
to the state of the art, namely an array with a randomly chosen combining matrix and a sparse array with optimized sensor positions. 
Furthermore, the compressive array is not only superior to the random and sparse arrays with respect to its estimation capabilities but also in terms of its ability to alter its weights on demand and thus facilitate signal-adaptive measurements.
The comparison between the proposed designs demonstrates a trade-off between the minimization of the CRB and the increase in the sidelobe level. In the proposed design, the trade-off ``CRB vs. sidelobe level'' can be controlled by setting the parameters during the optimization. This provides an additional degree of freedom for the system design that is unavailable in case of random and sparse arrays.


It is worth mentioning that similar efforts in spatial domain processing exist in the context of beam space array processing \cite{LW:90,A:91,ZKS:93,XSRK:94,ZHM:96,G:98,HEGW:06,HK:88} and hybrid beamforming \cite{VV:10,HGB:10,NBH:10}. In contrast to the element space processing, where signals derived from each element are weighted and summed to produce the array output, the beam space processing is a two-stage scheme. The first stage takes the array signals as an input and produces a set of multiple outputs, which are then weighted and combined to form the array output. These multiple outputs may be thought of as the output of multiple beams. 
The weights applied to different beam outputs are finally optimized according to a specific optimization criterion \cite{G:04}. 
%
In hybrid beamforming, the main idea is to apply beamforming and precoding techniques in both, the radio-frequency (RF) and the baseband (BB) \cite{BL:14}. This technique has attracted significant research attention in millimeter wave (mmWave) applications \cite{PK:11} for the next-generation indoor and mobile wireless networks \cite{ARAPH:14,RSPLLKCA:14}. 
While the overall goal in these areas is similar (reducing the number of digitally processed receiver channels), the actual design criterion for  the antenna is entirely different from the one we consider in this paper. We aim to obtain an array that is ideally suited for DOA estimation in the sense that it achieves an accurate estimate (by minimizing the CRB) while controlling the sidelobe characteristics by a prespecified probability of detecting a false direction.

This manuscript is organized as follows. Section~\ref{sec_dm} introduces the data model for the compressive arrays we consider and discusses the impact of the different sources of noise. The proposed design is shown in Section~\ref{sec_design}, where we discuss the general approach as well as the two specific design methods based on the spatial correlation function and the CRB, respectively. Section~\ref{sec_design} also contains the derivation of the analytical expression for the probability to detect a false peak. A discussion is contained in Section~\ref{sec:disc},  covering various possible extensions and an analysis of the achievable estimation accuracy compared to sparse arrays. Numerical results are presented in Section~\ref{sec_sims} before concluding in Section~\ref{sec_concl}.

\begin{figure*}[tbh]
\vspace*{-\baselineskip}
    \centering
    \includegraphics[width=0.9\textwidth]{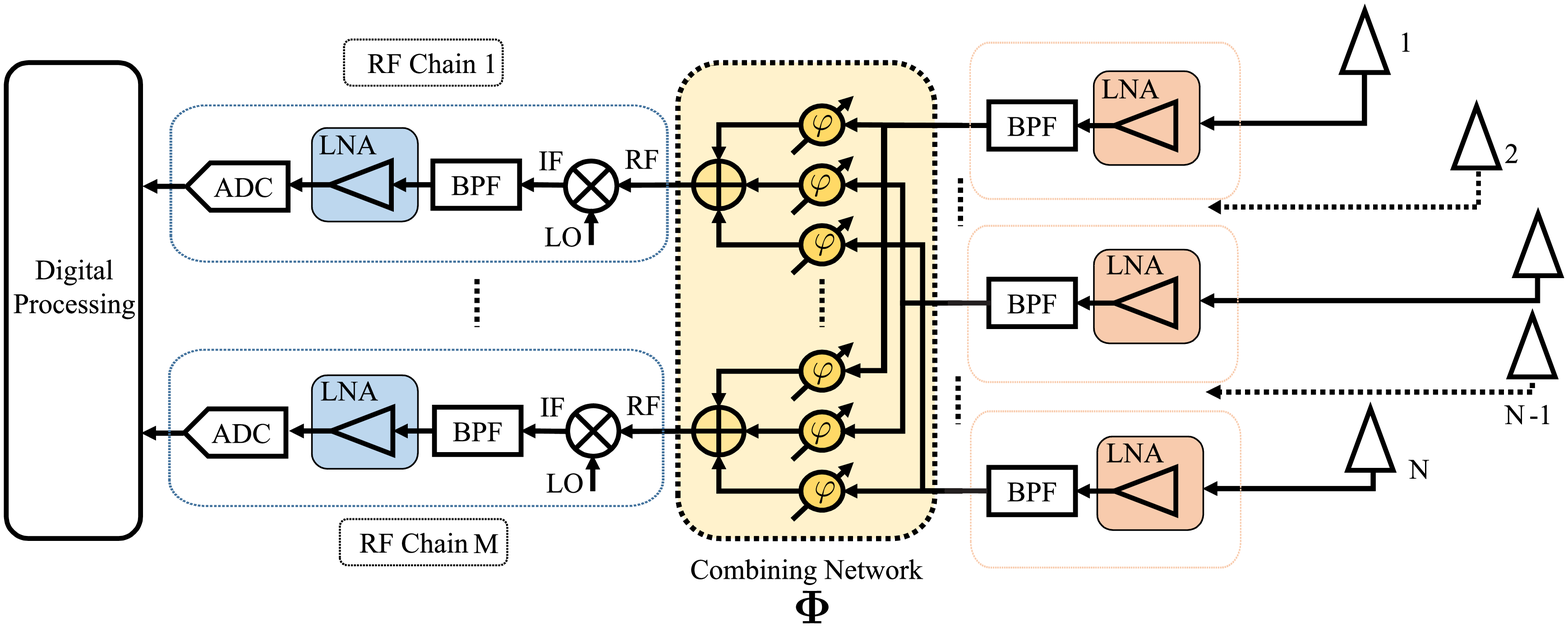}
    \caption{Compressive array hardware architecture}
    \label{fig:SystemModel}
    \vspace*{-0.2cm}
\end{figure*}

\section{System model for compressive arrays}\label{sec_dm}

\subsection{Data model for narrowband DOA estimation}
Consider $K$ narrowband plane waves impinging on an array of $N$ antenna elements. At the antenna output, the received (baseband) signal can be expressed as
\begin{equation}
  \bm{y}(t) =  \sum_{k = 1}^K \bm{a}(\bm{\gamma}_k) \cdot s_k(t) + \bm{n}(t), 
  \label{eqn:nocs_doa}
\end{equation}
where $\bm{y}(t)\in \mathbb{C}^{N \times 1}$ is a vector of antenna outputs, $\bm{n}(t) \in \mathbb{C}^{N \times 1}$ is an  additive noise vector, $t$ indicates the continuous time, and $\bm{a}(\bm{\gamma})$ denotes the antenna response as a function of the parameter vector $\bm{\gamma}^{\rm T}= [\theta, \psi, \bm{p}^{\rm T}]$ with $\theta$ and $\psi$  being the azimuth and the elevation angles, while $\bm{p} \in \mathbb{C}^{2 \times 1}$ represents the Jones vector that describes the polarization state of the incident plain wave at the receiver. Additionally, $s_k(t)$ in \eqref{eqn:nocs_doa} denotes the amplitude of the $k$th source, whereas $\bm{\gamma}^{\rm T}_k = [\theta_k, \psi_k, \bm{p}_k^{\rm T}]$  is the vector containing its azimuth ($\theta_k$)  and elevation ($\psi_k$) angles of arrival along with its Jones vector $\bm{p}_k^{\rm T} = [p_{k, 1}, p_{k,2}]$.
It is often useful to write \eqref{eqn:nocs_doa} in a matrix from as
\begin{equation}
   \bm{y}(t) = \bm{A} \cdot \bm{s}(t) + \bm{n}(t).
   \label{eqn:nocs_mtrdoa}
\end{equation}
Here, $\bm{A} = [\bm{a}(\bm{\gamma}_1), \bm{a}(\bm{\gamma}_2), \cdots, \bm{a}(\bm{\gamma}_K)] \in \mathbb{C}^{N \times K}$ is the array steering matrix and $\bm{s}(t) = [s_1(t), s_2(t), \cdots, s_K(t)]^{\rm T}\in \mathbb{C}^{K \times 1}$ is a vector containing the complex amplitudes of the $K$ sources. 


\subsection{Compressive arrays}
\label{sec:CSArray}

The model in \eqref{eqn:nocs_mtrdoa} presumes a dedicated
radio frequency (RF) receiver chain for each individual antenna element including a low-noise amplifier (LNA),
filters, down-conversion, analog-to-digital (ADC) conversion, etc. For specific applications, however, such separate RF chains
for each antenna element may come at a high cost in terms of the overall receiver complexity, the amount of data to be processed in the digital domain (e.g. FPGA) and power consumption. In order to reduce the number
of RF channels without a loss in the array aperture, we apply the \textit{compressive}
approach, where the antenna outputs are first linearly
combined in the analog domain and then passed through a lower
number  of RF chains to obtain the digital baseband signals as illustrated in Figure~\ref{fig:SystemModel}.
In this way, $M$ RF receiver channels (fewer than the $N$ antenna elements) are used
for signal processing in the digital domain. 

%
 The signal combining can be done at different stages within the receiver, e.g., on the RF (Radio Frequency) signal or at the IF (Intermediate Frequency) stage. The particular choice on where to place the combining network highly depends on the application, especially the considered frequency. In any case, additional signal losses will be introduced by the power splitters and combiners as well as the phase shifters inside the combining network. The actual losses' value will depend on multiple parameters including frequency, bandwidth and adaptability of the phase shifters. However, these losses need to be compensated by LNAs placed in each receiver chain as shown in Figure~\ref{fig:SystemModel}.

To this end, let $\bm{\Phi} \in \compl^{M \times N}$ denote the analog combining matrix
of a compressive array which compresses the output of $N$ antenna elements to $M$ active RF channels.
Then, the complex (baseband) antenna output \eqref{eqn:nocs_mtrdoa} after combining can be expressed as
\begin{align}
   \tilde{\bm{y}}(t) = \bm{\Phi}\left( \bm{A} \cdot \bm{s}(t) + \bm{v}(t) \right) + \bm{w}(t),
	\label{eqn:nm_doa}
\end{align}
%
where $[\bm{\Phi}]_{m,n} = \alpha_{m,n}e^{\jmath \varphi_{m,n}}, \alpha_{m,n} \in [0,1], \varphi_{m,n} \in [0, 2\pi], m = 1, 2,\cdots M, n = 1,2, \cdots, N$ \footnote{Note that $\bm{\Phi}$ does not need to be fully meshed, i.e., we do not need to connect each of the $N$ antennas to each of the $M$ outputs. In case the $n^{\rm{th}}$ antenna is not connected to the $m^{\rm{th}}$ RF channel the corresponding entry in the combining matrix is set to zero, e.g., $[\bm{\Phi}]_{m,n} = 0$.}, whereas $\bm{v}(t) \in \mathbb{C}^{N \times 1}$ and $\bm{w}(t) \in  \mathbb{C}^{M \times 1}$ are noise vectors with covariances $\bm{R}_{\rm vv}$ and $\bm{R}_{\rm ww}$ that represent additive noise sources  which act before and after\footnote{In the CS community, the former is often referred to  as ``signal noise'' and the latter as ``measurement noise''.} the combining network, respectively. 
For example, LNAs placed ahead the combining network contribute to $\bm{v}(t)$ (signal noise), whereas the ones placed behind the combining network contribute to $\bm{w}(t)$ (measurement noise).

Let $\tilde{\bm{A}} = \bm{\Phi} \cdot \bm{A}$ be the effective array steering matrix after combining, then \eqref{eqn:nm_doa} becomes
\begin{equation}
   \tilde{\bm{y}}(t) = \tilde{\bm{A}} \cdot \bm{s}(t) +  \tilde{\bm{n}}(t),
   \label{eqn:dm_doa}
\end{equation}
where
\begin{equation}
   \tilde{\bm{n}}(t) = \bm{\Phi} \cdot \bm{v}(t) + \bm{w}(t)
   \label{eqn:noise}
\end{equation}
is the effective noise vector with covariance $\bm{R}_{\rm{nn}} = \bm{\Phi} \bm{R}_{\rm vv} \bm{\Phi}^\herm + \bm{R}_{\rm ww}$. Assuming that $\bm{v}(t)$ and $\bm{w}(t)$ are white with elements that have variance $\sigma_1^2$ and $\sigma_2^2$, respectively, the covariance of $\bm{\tilde{n}}(t)$ becomes 
 $\bm{R}_{\rm{nn}}=\sigma_1^2 \bm{\Phi} \bm{\Phi}^\herm + \sigma_2^2 \bm{I}$.
 
Given \eqref{eqn:dm_doa}, we aim to design $\bm{\Phi}$ in such a way that it allows for a robust and efficient estimation of the DOAs of the $K$ sources $s_k(t)$ from the set of measurements $\tilde{\bm{y}}(t)$. Hence, our main design goal includes the minimization of the number of the receiver chains while providing a minimum variance of the DOA estimates and a reduced probability of spurious and/or ghost path estimates.

\section{Design of the combining matrix}\label{sec_design}

\subsection{Generic design approach}
Consider the receiver architecture from Fig. \ref{fig:SystemModel} where 
 the combining network is realized by: (i) splitting the analog RF signal of each of the $N$ antennas into $L\leq M$ branches; (ii) applying phase shifts in each of the branches; (iii) adding the branches to form each of the $M$ outputs, which are then passed to the $M$ RF chains. Mathematically, 
 we model this structure by a matrix $\bm{\Phi}$ with elements given by
 \begin{equation}
     [\bm{\Phi}]_{m,n} = \begin{cases}
      \frac{1}{\sqrt{L}} \cdot \eta \cdot {\rm e}^{\jmath \varphi_{m,n}} & \mbox{if $(n,m)$ are connected} \\
       0 & \mbox{otherwise},
    \end{cases}
    \label{eqn:phi_elements}
 \end{equation}
 where the connections between antennas and ports are such that $\bm{\Phi}$ has $L$ nonzero elements per column. In~\eqref{eqn:phi_elements}, the factor $\frac{1}{\sqrt{L}}$ represents the power splitting of each antenna's signal to $L$ branches and $\eta \in (0,1]$ is a scalar parameter that attributes for the fact that each analog branch (consisting of a power splitter, a phase shifter, and a combiner) is non-ideal and incorporates losses. A loss-less combining network would correspond to the special case $\eta=1$.
 
From  \eqref{eqn:phi_elements}, the combining matrix $\bm{\Phi}$ has $MN$ elements that provide $MN$ degrees of freedom for its design. In the CS literature, a typical approach for choosing $\bm{\Phi}$ would be to draw $\varphi_{m,n}$ randomly.  This, however, gives little control over the array characteristics. Furthermore, it might result in unwanted effects as high sidelobes and blind spots \cite{IRD:15}.

Here, we aim at a design of $\bm{\Phi}$  that results in an effective array
that has desired properties depending on the application scenario, e.g., uniform sensitivity and low cross-correlation for direction finding, adaptive spatial selectivity for parameter estimation during beam tracking, etc. 
 Generally, the design task can be formulated as the following constrained optimization problem
\begin{align}
    \bm{\Phi}_{\rm opt} = \arg \underset{\bm{\Phi}}{\min} J(\bm{\Phi}) \text{ s.t. } c(\bm{\Phi}, \alpha, \beta, \cdots),
    \label{eqn:opt_tsk_gnrl}
\end{align}
where $J(\bm{\Phi})$ is some objective function defined by the scenario and $c(\bm{\Phi}, \alpha, \beta, \cdots )$ represents the set of optimization constraints. In the following, we propose two particular formulations of \eqref{eqn:opt_tsk_gnrl} for direction finding applications: based on the spatial correlation function (SCF) and the Cram\'{e}r-Rao Lower Bound (CRB).

For the SCF-based approach, we build our design on the spatial correlation function 
defined as
\begin{equation}
     \rho(\bm{\gamma}_1, \bm{\gamma}_2) = \tilde{\bm{a}}(\bm{\gamma_1})^{\rm H}\cdot \ma{\tilde{a}}(\bm{\gamma}_2),
\end{equation}
where
  $  \ma{\tilde{a}}(\bm{\gamma}) = \ma{\Phi} \cdot \ma{{a}}(\bm{\gamma})$ presents the effective array manifold after combining.
The main idea is to design $\bm{\Phi}$ such that the spatial
correlation function $\rho(\bm{\gamma}_1, \bm{\gamma}_2)$ follows as close
as possible some pre-specified target $T(\bm{\gamma}_1, \bm{\gamma}_2)$. By defining an appropriate target function, we can provide desired properties in the spatial correlation function as discussed in the following.

Although SCF is often useful for gaining initial insight into the array's estimation capabilities, its ability to provide quantitative  evaluation of the achievable estimation quality is limited, especially in the case of multiple sources. Therefore, we propose a second approach that is based on the specific requirements on the estimation accuracy. More specifically, it aims at improving the accuracy of DOA estimation by designing $\bm{\Phi}$ such that it minimizes the CRB while keeping the probability of detecting a false direction at a certain (desired) level.

\subsection{Design based on the SCF}
\label{sec:SCF_design}

For the sake of simplicity, in the remainder of the paper it is assumed that the sources are located in the azimuthal plane of the antenna array and have an identical polarization of the impinging waves with perfectly matched antennas. Hence, the effective array manifold depends on the azimuth angle $\theta$ only, i.e., $ \ma{\tilde{a}}(\bm{\gamma}_1) =  \ma{\tilde{a}}(\theta_1)$. Note that an extension to a more general case is straightforward and is sketched in Section \ref{sec:disc_extns}.

Under these assumptions, an ideal generic array for
direction finding would satisfy the conditions
\begin{align}
  \rho(\theta_1, \theta_2) = \ma{\tilde{a}}(\theta_1)^\herm \cdot \ma{\tilde{a}}(\theta_2) 
		= \begin{cases}
		      {\rm const} & \theta_1 = \theta_2  \\
					0 & \theta_1 \neq \theta_2
		   \end{cases}. \label{eqn:manifold_ideal}
\end{align}
The first condition guarantees that the array gain is constant over all azimuth angles
and makes the array uniformly sensitive, whereas the second condition forces optimal cross-correlation properties to tell signals from different directions apart.
However, this is an example for a generic direction finder. For particular applications,
the design goal may differ, i.e., constraining on a certain sector of angles only or allowing
certain values for the residual cross-correlation. We denote the target function as $T(\theta_1,\theta_2)$,
where $T(\theta_1,\theta_2) = {\rm const} \cdot \delta(\theta_1 - \theta_2)$ represents the ideal generic array~\eqref{eqn:manifold_ideal}.

Due to the finite aperture of an $N$-element array, the target in \eqref{eqn:manifold_ideal} can only be achieved approximately. This allows us to define a criterion for the optimization of $\ma{\Phi}$ according to the cost function
\begin{align}
   e(\ma{\Phi},\theta_1,\theta_2) & = \left| \ma{\tilde{a}}(\theta_1)^\herm \cdot \ma{\tilde{a}}(\theta_2)  - T(\theta_1,\theta_2)\right|  \label{eqn:def_error}\\
	& = \left| \ma{{a}}(\theta_1)^\herm \cdot \ma{\Phi}^\herm \cdot \ma{\Phi} \cdot\ma{{a}}(\theta_2)  - T(\theta_1,\theta_2)\right|. \notag 
\end{align}
We can approximate the continuous variables $\theta_1$ and $\theta_2$ by considering the $P$-point sampling grid $\theta_p^{\rm (G)}, p=1, 2, \ldots, P$ 
and define the $P \times P$ matrices $\ma{E}$ and $\ma{T}$
according to $\ma{E}_{(i,j)} =  e(\ma{\Phi},\theta_{i}^{\rm (G)},\theta_{j}^{\rm (G)})$
and $\ma{T}_{(i,j)} =  T(\theta_{i}^{\rm (G)},\theta_{j}^{\rm (G)})$.
After insertion into \eqref{eqn:def_error} we obtain
\begin{align}
   \ma{E} = \left|\ma{A}^\herm \cdot \ma{\Phi}^\herm \cdot \ma{\Phi} \cdot \ma{A} - \ma{T}\right|. \label{eqn:error_matrix}
\end{align}
Based on \eqref{eqn:error_matrix}, the quality of $\ma{\Phi}$ can be assessed by a suitable norm of $\ma{E}$.
As a first step, let us consider the Frobenius norm, i.e., we optimize
$\ma{\Phi}$ according to
%
\begin{align}
   \ma{\Phi}_{\rm opt} = \mathop{\arg\min}_{\ma{\Phi}} \left\|\ma{E} \right\|_{\rm F}^2.
	\label{eqn:cf_main}
\end{align}
%
In the special case\footnote{This
condition is, e.g., fulfilled for an ULA if the sampling grid is chosen
to be uniform in the spatial frequencies (direction cosines).
Moreover, for many arrays the condition is approximately fulfilled (e.g.,
for UCAs).
In this case, the closed-form solution can still be applied as a heuristic
method. } where $\bm{A} \cdot \bm{A}^\herm = C \cdot \bm{I}_M$, with $C$ being a constant,
the optimization problem in~\eqref{eqn:cf_main} admits a closed-form solution
as shown in the following theorem.

\begin{theorem}\label{thm:closedform}
		Let $\bm{S} = \bm{A} \cdot \bm{T} \cdot \bm{A}^\herm$
		and let $\bm{S}_M$ be a rank-$M$-truncated version of $\bm{S}$ obtained
		by setting its $P-M$ smallest eigenvalues to zero.
		Then the set of 
		optimal solutions to \eqref{eqn:cf_main} is given by the set of matrices $\ma{\Phi}$
		that satisfy $\bm{\Phi}^\herm \bm{\Phi} = \bm{S}_M$.
\end{theorem}
{\em Proof}: cf.~Appendix~\ref{app:proof_closedform}.

In other words, Theorem~\ref{thm:closedform} states that we can find an optimal~$\ma{\Phi}$ by computing a square-root factor of the best rank-$M$ approximation
of $\ma{S}$. Moreover, the following corollary can be found from Theorem~\ref{thm:closedform}:
 
\begin{cor}\label{cor:roworth}
		Under the conditions of Theorem~\ref{thm:closedform} any matrix $\bm{\Phi}$ is optimal
		in terms of the ``ideal'' target from~\eqref{eqn:manifold_ideal} if and only if the rows
		of $\bm{\Phi}$ have equal norm and are mutually orthogonal.
\end{cor}
\begin{proof}
cf.~Appendix~\ref{app:proof_corollary}.
\end{proof}
Corollary~\ref{cor:roworth} agrees with the intuition that the measurements
(i.e., the rows of $\ma{\Phi}$) should be chosen such
that they are orthogonal in order to make every observation as informative as possible.
In addition, the corollary shows that this choice also minimizes 
$\fronorm{\bm{\Phi}^\herm \bm{\Phi} - C \cdot P \cdot \ma{I}_N}$
which contains the correlations between all pairs of columns in $\bm{\Phi}$ as well as
the deviation of the columns' norms (therefore, in a sense, this choice minimizes the
``average'' mutual correlation).
On the other hand, this also demonstrates that the optimization in \eqref{eqn:cf_main} 
is not sufficiently 
selective since all row-orthogonal matrices achieve the same minimum of the cost function.

The cost function \eqref{eqn:cf_main} assigns an equal weight to the error for all 
pairs of grid points $\theta_1^{\rm (G)}, \theta_2^{\rm (G)}$, i.e., it tries to 
maintain a constant main lobe with the same weight as it tries to 
minimize sidelobes everywhere. In practice it is often desirable to have more control
over the shape of the spatial correlation function, e.g., trading main lobe ripple against sidelobe levels
or allowing for a transition region between the mainlobes and sidelobes that is not
constrained. There are many ways such constraints could be incorporated, e.g., maximum
constraints on the magnitude of cross-correlation in some region and interval constraints
on the autocorrelation inside the mainlobe. For numerical tractability, we follow a simpler
approach by introducing a weighting matrix $\ma{W} \in \real^{P \times P}$ into~\eqref{eqn:cf_main}. 
The modified optimization problem is given by
\begin{align}
   \ma{\Phi}_{\rm opt} = \arg \mathop{\min}_{\ma{\Phi}} \left\|\ma{E} \odot \ma{W} \right\|_{\rm F}^2,
	\label{eqn:cf_main_w}
\end{align}
where $\odot$ represents the Schur (element-wise) product. The weighting matrix allows to put
more or less weight on the main diagonal (controlling how strictly the constant mainlobe
power shall be enforced), certain off-diagonal regions (controlling how strongly sidelobes in
these regions should be suppressed), or even placing zeros for regions that remain arbitrary
(such as transition regions between the mainlobe and the sidelobes). Thereby, more flexibility
is gained and the solution can be tuned to more specific requirements.

The drawback of~\eqref{eqn:cf_main_w} is that it does not admit a closed-form solution in general.
However, it can be solved by numerical optimization routines that are available in modern
technical computing languages. 

\subsection{Design based on the CRB}
\label{sec:design_crb}



For the case of a single source, 
a correlation-based DOA estimator amounts to finding the DOA $\theta_{0}$ that corresponds to the global maximum in the beamformer spectrum $D(\theta)$, i.e.,
\begin{align} 
\theta_{0} = \arg\underset{\theta}{\max} \; D(\theta) \equiv \arg\underset{\theta}{\max} \left(\displaystyle \frac{\tilde{\bm{a}}^{\rm H}(\theta)\cdot\bm{R}\cdot\tilde{\bm{a}}(\theta)}{\|\tilde{\bm{a}}(\theta)\|_2^2}\right)
, \label{fals-det-prob}
\end{align}
where 
$\bm{R} = \mathbb{E} \{\tilde{\bm{y}}(t) \cdot \tilde{\bm{y}}^{\rm H}(t) \}$ is the covariance matrix of the received signals. Note that in this case,  (\ref{fals-det-prob}) is equivalent to the maximum likelihood (ML) cost function, and therefore, the correlation-based DOA estimator is equivalent to the ML estimator. We define then the false detection as the event where the global maximum in the beamformer spectrum $D(\theta)$ is outside the mainlobe area (either $3$-dB or null-to-null beamwidth). The error probability $P_{d}$ is hence given by
\begin{align} P_{d} \equiv \text{Prob}\big(D(\theta_{0}) < D(\theta), \forall \theta \in \mathcal{U}\big), \label{def_prob}\end{align} 
where $\mathcal{U}$ denotes the set of DOAs corresponding to the directions in the beamformer spectrum outside the mainlobe area, and 
\begin{align}
D(\theta_0)  &= \left| \bm{\tilde{a}}_0^{\rm H}\bm{\tilde{y}}(t) \right |^2 
= \left|   \bm{\tilde{a}}_0^{\rm H}\bm{\tilde{a}}_0{s}(t) + \bm{\tilde{a}}_0^{\rm H}\tilde{\bm{n}}(t)  \right |^2, \nonumber \\
%
D(\theta)&  = \left| \bm{\tilde{a}}^{\rm H}\bm{\tilde{y}}(t) \right |^2 
= \left|   \bm{\tilde{a}}^{\rm H}\bm{\tilde{a}}_0{s}(t) + \bm{\tilde{a}}^{\rm H} \tilde{\bm{n}}(t) \right |^2
\label{eqn:beamf_directions}
\end{align}
with $\bm{\tilde{a}}_k \equiv \mathbf{\Phi}\cdot \bm{a}(\theta_k)$ and $\bm{\tilde{a}} \equiv \mathbf{\Phi}\cdot\bm{a}(\theta)$. 

A direct evaluation of (\ref{def_prob}) is analytically intractable. In order to proceed, we adopt the approach from \cite{FA:05} and approximate the continuous correlation function $b(\theta,\theta_{0}) = |\tilde{\bm{a}}^{\rm H}\tilde{\bm{a}}_{0}|/(\left\| \tilde{\bm{a}}\right\| \left\| \tilde{\bm{a}}_{0}\right\|) $ of the antenna array by its discretized version
\begin{align} b^{\rm d}(\theta,\theta_{0}) = \sum_{q=0}^{L} b(\theta,\theta_{0})\delta(\theta-\theta_{q}), \label{gl_discretized} \end{align}
where  $\theta_{0}\notin \mathcal{U}$ and $\theta_{q}\in \mathcal{U} \; \forall q>0$ are the directions corresponding to the mainlobe and the sidelobe peaks in the array's correlation function, respectively, whereas $L$ is the total number of the sidelobe peaks and $\delta(\theta)$ denotes the Dirac delta function. Using (\ref{gl_discretized}), the false detection probability can now be approximated by 
\begin{align} P_{d} \approx  \mbox{Prob}\bigg(\bigcup_{q=1}^{L} \big\{ D(\theta_{0})-D(\theta_{q})<0 \big\}\bigg).\label{err_prob}
\end{align} 
In order to simplify the calculation, we apply the union bound \cite{KR:05} on (\ref{err_prob}), and obtain
\begin{align} P_{d} \leq \sum_{q=1}^{L}\mbox{Prob}\big( D(\theta_{0})-D(\theta_{q})<0 \big). \label{gl_union_bound}\end{align} 
Now it remains to compute the individual probabilities $P_{q}\equiv \mbox{Prob}\big( D(\theta_{0})- D(\theta_{q}) < 0\big)$, $\forall q \in [1, L]$, where $P_{q}$ denotes the probability that the $q$-th sidelobe peak is higher than the mainlobe peak.

\begin{theorem}
    Suppose $\bm{v}(t)$ and $\bm{w}(t)$ are independent zero-mean complex white Gaussian noise vectors with covariances $\sigma_1^2\bm{I}$ and $\sigma_2^2  \bm{I}$, respectively and let $\Psi_{q}(s)$ be the moment generating function (MGF) of $(D(\theta_{0})- D(\theta_{q}))$ defined by \eqref{eqn_mfg} in  App. \ref{app:prooffalsedetecion}.  Then, $P_{q} = \mbox{Prob}\big( D(\theta_{0})- D(\theta_{q}) < 0\big)$ can be computed as
    \begin{equation}
        P_{q} \approx \frac{1}{2G} \sum_{g=1}^{G} \hat{\Psi}_q\Bigg( \frac{(2g-1)\pi}{2G}\Bigg)
	\label{eqn:false_prob},
    \end{equation}
    where 
     $\hat{\Psi}_q(\tau ) = (1-j \tan(\tau /2))\Psi_{q}\big(-s^{\rm p}_{q}(1+j \tan(\tau /2))\big)$, $s^{\rm p}_q$ is the so-called saddle point \cite{CH:86}, and $G$ is a natural number that determines the accuracy of the approximation.
    
    \label{thr:prblt}
\end{theorem}
\begin{proof}
    cf. App. \ref{app:prooffalsedetecion}.
\end{proof}

Applying Theorem \ref{thr:prblt} to \eqref{gl_union_bound}, we finally obtain
\begin{equation}
     P_{d} \leq \frac{1}{2G} \sum_{q=1}^{L} 
      \sum_{g=1}^{G} \hat{\Psi}_q\Bigg( \frac{(2g-1)\pi}{2G}\Bigg)
	\label{eqn:Pd_single_source}
\end{equation}

The analytic expression for the false detection probability can now be used to optimize the combining matrix $\bm{\Phi}$ with the objective to improve the DOA estimation accuracy. 
%
%
For detection of a single source, we can formulate it as 
\begin{align} 
 \ma{\Phi}_{\rm opt} = &\arg \underset{\bm{\Phi}}{\min}\max_{\theta_{0}}  \text{ CRB}(\bm{\Phi}, \theta_{0}) \label{opt_strat} \\\nonumber
&\mbox{s. t. }  P_{d}(\bm{\Phi}, \theta_{0}, \rho_{th}) < \epsilon_{0}, 
\label{eqn:CRB_optm}
%
\end{align} 
where $\text{ CRB}(\mathbf{\Phi}, \theta_{0})$ is given by the expression (\ref{single_doa}) in Appendix \ref{app:CRLB} and $\epsilon_0$ is the desired false detection level. The optimization strategy in (\ref{opt_strat}) aims to minimize the CRB and hence to improve the DOA estimation accuracy over the full angular range $(0, 2 \pi]$, while lowering the false detection probability to a desired value $\epsilon_{0}$ for a given SNR threshold point $\rho_{th}$. 

Both the constraint and the objective in (\ref{opt_strat}) are non-convex functions with respect to $(\mathbf{\Phi}, \theta_{0})$. The optimization problem is thereby a non-convex problem exhibiting a multi-modal cost function, where the optimal (global) solution can only be found by an exhaustive search strategy. Therefore, we apply a local minimizer to the above problem using an algorithm based on the interior-reflective Newton method \cite{CM:94}, \cite{CM:96}. However, by using this algorithm the obtained solution strongly depends on the initialization of the parameters $(\mathbf{\Phi}, \theta_{0})$. Moreover, there is no guarantee that the global optimum is found. 
One way of addressing this issue is to apply the algorithm several times, where for each run the initialization of the parameters $(\mathbf{\Phi}, \theta_{0})$ is different. In doing so, the obtained solution to (\ref{opt_strat}) is likely to be sufficiently close to the optimal solution. However, it might be time consuming due to the complexity of the optimization problem at hand. Another way of tackling this problem is to first obtain a solution for $\bm{\Phi}$ by the SCF approach described above and then use it for the initialization in (\ref{opt_strat}).


\section{Discussion}
\label{sec:disc}


\subsection{Multi-dimensional DOA}
\label{sec:disc_extns}
For simplicity, we have discussed the design of the compression matrices
only for 1-D case, i.e., the estimation of the azimuth angle, assuming
that all the sources are located on the plane of the array.
However, it is straightforward to generalize the design to the 2-D case where both, azimuth and elevation, are considered. For example, for the approach from Section~\ref{sec:SCF_design}
we can define the spatial correlation function in the 2-D case as
$\bm{\tilde{a}}(\theta_1,\psi_1)^\herm \bm{\tilde{a}}(\theta_2,\psi_2)$
and define a criterion to optimize $\bm{\Phi}$ via
\begin{align}
    e(\bm{\Phi},\theta_1,\theta_2,\psi_1,\psi_2)
    =& \left|\bm{\tilde{a}}(\theta_1,\psi_1)^\herm \bm{\tilde{a}}(\theta_2,\psi_2)\right. \notag \\
    & \left. - T(\theta_1,\theta_2,\psi_1,\psi_2)\right|.
\end{align}
Here $T(\theta_1,\theta_2,\psi_1,\psi_2)$ is the target SCF which
could for example be chosen as $T(\theta_1,\theta_2,\psi_1,\psi_2) = {\rm const}\cdot 
\delta(\theta_1-\theta_2)\cdot\delta(\psi_1-\psi_2)$.
To minimize the cost function 
we can introduce a $P_\theta$ by $P_\psi$
2-D sampling grid in 
azimuth and elevation and then align the sampled cost function
into a matrix $\bm{E}$ of size $P_\theta P_\psi \times P_\theta P_\psi$.

A similar extension is possible to incorporate the polarization of the incoming
wave. We can express the dependence of the array on the polarization of the incident
wave via
\begin{align}
    \bm{a}(\theta,\alpha,\phi)
    = \left[\bm{a}_{\rm H}(\theta), \, \bm{a}_{\rm V}(\theta)\right]\cdot
    \underbrace{\begin{bmatrix} \cos(\alpha) \\ \sin(\alpha)\cdot {\rm e}^{\jmath \phi}\end{bmatrix}}_{\bm{p}(\alpha,\phi)}
\end{align}
where $\bm{a}_{\rm H}(\theta)$ and $\bm{a}_{\rm V}(\theta)$ represent
the array response to a purely horizontal and a purely vertical incident plain wave, respectively,
and 
the parameters $\alpha,\phi$ describe the polarization state\footnote{Values of
$\phi=0$ and $\phi=\frac{\pi}{2}$ correspond to linear and circular polarized waves,
respectively. All other values of $\phi$ imply an elliptical polarization. The angle $\alpha$ describes the orientation of the polarization plane (e.g., $\phi=0$, $\alpha=0$
corresponds to a horizontal and $\alpha=\frac{\pi}{2}$ to a vertical polarized wave, respectively).} of the wave. 
We then can optimize $\bm{\Phi}$ by minimizing the error
$e(\theta_1,\alpha_1,\phi_1,\theta_2,\alpha_2,\phi_2) = 
|\bm{a}^\herm(\theta_1,\alpha_1,\phi_1)\cdot \bm{a}(\theta_2,\alpha_2,\phi_2)
- T(\theta_1,\theta_2)|$. Note that the target does not depend on the polarization
parameters since our goal is to achieve a high separation in the angular domain
for any polarization (and not a separation in the polarization domain).
This cost function can be minimized by defining a multidimensional grid
which leads to a corresponding error matrix $\bm{E}$ of size 
$P_\theta P_\alpha P_\varphi \times P_\theta P_\alpha P_\varphi$.

The same extensions are also possible to the CRB-based approach from Section~\ref{sec:design_crb}. 
For example, we can  incorporate both azimuth and elevation by defining the 2-D beamformer spectrum as
\begin{equation}
     D(\theta, \psi) = \bm{\tilde{a}}(\theta,\psi)^\herm \cdot \bm{R} \cdot \bm{\tilde{a}}(\theta,\psi),
\end{equation}
and writing \eqref{gl_union_bound} as
\begin{align} 
P_{d} \leq \sum_{q=1}^{\bar{L}}\mbox{Prob}\big( D(\theta_{0}, \psi_0)-D(\theta_{q}, \psi_q)<0 \big),
\label{gl_union_bound_ext}
\end{align}
where ${L}$ is now the total number of sidelobe peaks in the discretized 2-D correlation function. The expression for the individual probabilities \mbox{$P_q = \mbox{Prob}\big( D(\theta_{0}, \psi_0)-D(\theta_{q}, \psi_q)<0 \big)$} remains the same as in \eqref{eqn:false_prob} where the MGF is calculated for $D(\theta_{0}, \psi_0)-D(\theta_{q}, \psi_q)$ instead of $D(\theta_{0})-D(\theta_{q})$. The same holds for the total false detection probability given by \eqref{eqn:Pd_single_source}.

\subsection{Arbitrary number of sources} 
\label{sec:multi}

So far, we have discussed the case of a single source's wave impinging on the antenna array. However,  we can easily extend the CRB design presented in Section \ref{sec:design_crb} to account for the presence of multiple signal sources by applying in \eqref{eqn:CRB_optm} the full CRB given by \eqref{eqn:CRLB} and modifying the false detection probability expression in \eqref{eqn:Pd_single_source}. Particularly, assuming a correlation-based DOA estimator,  we need to compute the probability that one strongest source is falsely detected in the presence of $K-1$ weaker ones. In light of that, \eqref{eqn:beamf_directions} becomes
\begin{align}
D(\theta_0)  &= \left| \bm{\tilde{a}}_0^{\rm H}\bm{\tilde{y}}(t) \right |^2 
= \left|   \sum_{k=0}^{K-1}\bm{\tilde{a}}_0^{\rm H}\bm{\tilde{a}}_k{s}_k(t) + \bm{\tilde{a}}_0^{\rm H}\tilde{\bm{n}}(t)  \right |^2, \nonumber \\
D(\theta_q)&  = \left| \bm{\tilde{a}}_q^{\rm H}\bm{\tilde{y}}(t) \right |^2 
= \left|   \sum_{k=0}^{K-1}\bm{\tilde{a}}_q^{\rm H}\bm{\tilde{a}}_k{s}_k(t) + \bm{\tilde{a}}_q^{\rm H} \tilde{\bm{n}}(t) \right |^2
\label{eqn:beamf_directions_K}
\end{align}
and $P_{q} = \text{Prob} \big( D(\theta_{0})- D(\theta_{q}) < 0\big)$ can again be calculated  by
\begin{equation}
P_{q} \approx \frac{1}{2G} \sum_{g=1}^{G} \hat{\Psi}_q\Bigg( \frac{(2g-1)\pi}{2G}\Bigg).
	\label{eqn:false_prob_K}
\end{equation}
The difference between \eqref{eqn:false_prob_K} and \eqref{eqn:false_prob} is that in case of multiple sources the vector $\bm{r}$ in \eqref{eqn:mu} in Appendix \ref{app:prooffalsedetecion} becomes equal to $ \sum_{k=0}^{K-1}\bm{\tilde{a}}_{k} s_{k}$ when calculating the non-centrality parameters of the MGF.


\subsection{Adaptive design}
\label{sec:adptv_dsn}

In \eqref{eqn:opt_tsk_gnrl}, our target is a static combining matrix  that
yields an array with certain properties, such as uniform sensitivity and low sidelobe level, which is a good choice if no prior knowledge of the targeted sources is available. However, we can extend this approach towards an adaptive design that makes use of
the fact that for a slowly changing scene, the estimates from the previous
snapshots provide prior information about the source locations
in the next snapshots. This fact can be utilized for adaptive focusing of the array’s sensitivity towards regions of interest where the targets are expected \cite{IRG:15}.
In doing so, the SNR and the effective resolution in these directions of interest can be further improved, resulting in a superior DOA estimation performance.
To achieve this, we adopt a sequential measurement strategy which
starts with a combining matrix designed for uniform sensitivity and then gradually refine it towards the directions of interest that have been identified in the observations collected so far \cite{IRG:15}.

The adaptation mechanism proceeds as follows:
\begin{enumerate}
	\item We begin by scanning the scene with a matrix $\ma{\Phi}$ designed according to \eqref{eqn:cf_main} or \eqref{eqn:cf_main_w}, designed for a uniform target $T$ or according to \eqref{eqn:CRB_optm} for the full angular range $\theta_0 \in (0, 2\pi]$.
	\item Identify regions of interest based on an estimate of the angular power spectrum
	     obtained from the initial observation(s).
	\item Define a focusing region $\Theta$ as the union of all regions of interest.
	\item Modify $\ma{\Phi}$ by solving \eqref{eqn:cf_main} or \eqref{eqn:cf_main_w} for a target designed for the focusing region $\Theta$ in the SCF-based approach or solving \eqref{eqn:CRB_optm} with a restricted angular range.
	\item As the sources are assumed to change their position gradually, track sources by repeating steps (2) to (4)
	      sequentially, moving the regions of interest along with the currently identified source locations.
	\item Every $S$ snapshots, rescan the scene with a matrix $\ma{\Phi}$ designed for a uniform sensitivity
	      in order to detect newly appearing sources. If new sources are found, incorporate their location into the set $\Theta$.
\end{enumerate}
The parameter $S$ represents a design parameter that determines how quickly the system reacts to sources appearing outside the current direction of interest.
Note that this adaptation mechanism allows for many degrees of freedom, e.g., in terms of the rate of adaptation of
$\ma{\Phi}$ or the definition of the focusing regions. Some results on the performance of such an adaptive design based on an example of the SCF optimized combining matrix can be found in \cite{IRG:15}.

\subsection{Estimation quality}

In this section, we provide an analysis of the achievable performance of the proposed compressive arrays for DOA estimation.
In the noisy case, there are two main estimation quality measures: the achievable estimation accuracy, and the resolution capabilities.

\subsubsection{Estimation accuracy}
For a fixed aperture, the achievable accuracy is mainly determined by the SNR at the output of the antennas. For this reason, we compare the SNR of a compressive array and a sparse array at the same number of active channels $M$.

%
%

 We express the output signal for a single source via~\eqref{eqn:nm_doa} as
\begin{align}
   \bm{\tilde{y}}(t) = \bm{\Phi} \cdot \bm{a}(\bm{\gamma}_1)\cdot s_1(t) + \bm{\tilde{n}}(t),
\end{align}
where the elements of $\bm{\Phi}$ are given by~\eqref{eqn:phi_elements}
and the covariance of $\bm{\tilde{n}}(t)$ is given by $\bm{R}_{\rm nn} = \sigma_1^2\bm{\Phi}\bm{\Phi}^\herm +
\sigma_2^2  \bm{I}_M$. 
The SNR of the compressed array can then be computed as
\begin{align}
   \rho_{\rm c} & = \frac{\expvof{\left\|\bm{\Phi} \cdot \bm{a}(\bm{\gamma}_1)\cdot s_1(t)\right\|^2}}
               {{\traceof{\bm{R}_{\rm nn}}}} \notag \\ 
       & = \frac{\left\|\bm{\Phi} \cdot \bm{a}(\bm{\gamma}_1)\right\|^2 \cdot P_{\rm s}}{\traceof{\bm{\Phi}^\herm \bm{\Phi}}\sigma_1^2  + M\sigma_2^2} \notag \\
        & = \frac{
        \traceof{\bm{\Phi} \cdot \bm{a}(\bm{\gamma}_1) \cdot \bm{a}^\herm(\bm{\gamma}_1) \cdot \bm{\Phi}^\herm} 
        \cdot P_{\rm s}}
        {\traceof{\bm{\Phi}^\herm \bm{\Phi}}\sigma_1^2 + M \sigma_2^2},
        \label{eqn:snr_gamma}
\end{align}
where $P_{\rm s} = \expvof{|s_1(t)|^2}$ is the source power. 

As evident from \eqref{eqn:snr_gamma}, the SNR is dependent on the parameter vector $\bm{\gamma}$,
i.e., on the DOA. It is therefore meaningful to consider the average SNR over all possible
source directions. This requires to compute the average of
$g(\bm{\gamma}) = \left\|\bm{\Phi} \cdot \bm{a}(\bm{\gamma})\right\|^2$ over $\bm{\gamma}$ which is not possible
without further assumptions either about the array or about $\bm{\Phi}$.
Let $\bar{g}$ be the average of $g(\bm{\gamma})$ over $\bm{\gamma}$, i.e.,
$\bar{g} = \Gamma^{-1} \int g(\bm{\gamma}) {\rm d} \bm{\gamma}$
with $\Gamma = {\int 1 {\rm d}{\bm{\gamma}}}$.
Moreover, let us define the matrix $\bm{J} = \Gamma^{-1} \int \bm{a}(\bm{\gamma}) \cdot \bm{a}(\bm{\gamma})^\herm {\rm d}{\bm{\gamma}}$ so that $\bar{g} = \traceof{\bm{\Phi} \cdot \bm{J} \cdot \bm{\Phi}^\herm}$.
To proceed, we would like to replace $\bar{g}$ by $\traceof{\bm{\Phi}  \bm{\Phi}^\herm}$. 
We can always do so when $\bm{J}=\bm{I}_M$ which implies that the beam patterns of all
antennas are orthogonal over the entire parameter space. This is, e.g., fulfilled for an ULA
if it is parametrized by spatial frequencies $\mu=\cos(\theta)$. 
Furthermore, for $\bm{J}\neq \bm{I}_M$ one can show that 
$\expvofsub{\bm{\Phi}}{\bar{g}} = \expvofsub{\bm{\Phi}}{\traceof{\bm{\Phi}\bm{\Phi}^\herm}}$
for any random ensemble of $\bm{\Phi}$ where its elements are i.i.d. Note that in our case,
due to \eqref{eqn:phi_elements}, 
$\traceof{\bm{\Phi} \bm{\Phi}^\herm}$ is not random but deterministic. Hence, the expectation
on the right-hand side is not needed.
In light of this assumption, 
we can express the average SNR $\bar{\rho}_{\rm c}$ (averaged over $\bm{\gamma}$)
as
\begin{align}
\bar{\rho}_{\rm c} = 
   \frac{\traceof{\bm{\Phi} \bm{\Phi}^\herm} \cdot P_{\rm s}}{\traceof{\bm{\Phi}^\herm \bm{\Phi}} \sigma_1^2+ M \sigma_2^2}
   = \frac{   \left\|\bm{\Phi}\right\|_{\rm F}^2 \cdot P_{\rm s}}{\left\|\bm{\Phi}\right\|_{\rm F}^2
\sigma_1^2 + M \sigma_2^2}
\end{align}
Using~\eqref{eqn:phi_elements} it is easy to see that 
$\left\|\bm{\Phi}\right\|_{\rm F}^2 = \left(\frac{\eta}{\sqrt{L}}\right)^2 \cdot N \cdot L = \eta^2\cdot M$. Therefore, the average SNR becomes
\begin{align}
\bar{\rho}_{\rm c}  
   = \frac{\eta^2 \cdot N \cdot P_{\rm s}}{\eta^2 \cdot N \sigma_1^2 + M \sigma_2^2}
   = \frac{P_{\rm S}}{\sigma_1^2 + \sigma_2^2} \cdot
     \frac{\eta^2 \cdot N}{\frac{\eta^2 \cdot N}{1+\frac{\sigma_1^2}{\sigma_2^2}}
                          +\frac{M             }{1+\frac{\sigma_2^2}{\sigma_1^2}} }
\end{align}
To compare this SNR to the one that can be achieved with a sparse array we model
the observed signal as $\bm{a}_{\rm s}(\bm{\gamma}_1) \cdot s_1(t) + \bm{w}_{\rm s}(t)$,
where,
to make the comparison fair, 
the elements of $\bm{w}_{\rm s}(t)$ are i.i.d.~with
 variance $\sigma_1^2+ \sigma_2^2$.
We then obtain for the SNR of a sparse
array
\begin{align}
  \bar{\rho}_{\rm s} = \frac{M \cdot P_{\rm s}}{M  (\sigma_1^2+ \sigma_2^2)}  
  = \frac{ P_{\rm s}}{\sigma_1^2+ \sigma_2^2}.
\end{align}
Therefore, the ratio of the SNRs becomes
\begin{align}
  \frac{\rho_{\rm c}}{\rho_{\rm s}} =
   \frac{\eta^2 \cdot N}{\frac{\eta^2 \cdot N}{1+\frac{\sigma_1^2}{\sigma_2^2}}
                        +\frac{M             }{1+\frac{\sigma_2^2}{\sigma_1^2}} }
\end{align}
Overall, this shows that for dominating signal noise (i.e., $\sigma_1^2 \gg \sigma_2^2$)
we have $\rho_{\rm c} \approx \rho_{\rm s}$ and thus there is no SNR gain from
using the compressed arrays. On the other hand, for dominating measurement noise
(i.e., $ \sigma_2^2 \gg \sigma_1^2$), the SNR ratio approaches $\eta^2 \frac{N}{M}$ which
means an SNR improvement if the efficiency of the lossy components satisfies $\eta > \sqrt{\frac{N}{M}}$.
In practice, the compression ratio $\frac{N}{M}$ can be quite high and therefore, the
SNR improvement of the compressive arrays can be very significant.




\subsubsection{Resolution}
The ability to distinguish closely spaced sources is an important characteristic of an antenna array. The achievable resolution of the array mainly depends on its aperture, i.e., the largest distance between pairs of antenna elements. For ULAs, the aperture is equal to $(N-1)\lambda/2$ since the elements are spaced in half-wavelength distance from each other. For compressive arrays as well as sparse arrays, this distance can be increased further. As a result, the array's correlation function becomes sharper, at the price of an increase in sidelobes (grating lobes). We can control the height of the grating lobes by proper design of the antenna placement (in the case of sparse arrays) as well as the analog combining network (in the case of compressive arrays). In general, we expect that at the same covered aperture and the same number of active RF chains, the compressive arrays will have lower sidelobes (since the many degrees of freedom in the analog combining network allow to suppress the sidelobes significantly). As it is difficult to quantify the achievable sidelobe suppression analytically, we will focus on this aspect in the numerical results in Section~\ref{sec_sims}.

\section{Numerical results}\label{sec_sims}
In this section, we evaluate the performance of the compressive array with optimized combining network and compare it to its closest counterparts in terms of the aperture and hardware complexity, namely random and sparse arrays. We perform the numerical study based on a uniform circular array (UCA) with $N = 9$ elements that are compressed to $M=5$ receiver channels (this amounts to $\approx 1.8$ times reduction in the number of receiver channels). Note that for an UCA with isotropic elements the response of the $n^{\rm{th}}$ antenna element as a function of the azimuth angle $\theta$ can be written as
\begin{equation}
a_n(\theta) = e^{\jmath 2 \pi \tilde{R} \cos(\theta - \vartheta_n)},
\end{equation} 
where 
$\vartheta_n = 2 \pi (n-1)/N$ with $n = 1, 2, \cdots, N$ and $\tilde{R}= \frac{R}{\lambda}$ is the array radius normalized to the wavelength. For both proposed designs, the radius $\tilde{R}$ was fixed and set to $0.65$. 

The combining matrix $\ma{\Phi}$ is chosen according to
$[\ma{\Phi}]_{(m,n)} = {\rm e}^{\jmath \varphi_{m,n}}$, where $\varphi_{m,n}$
are the optimization variables in the proposed approaches. To find an optimized design $\ma{\Phi}_{\rm opt}$ we solve
the weighted optimization problems~\eqref{eqn:cf_main_w} and \eqref{eqn:CRB_optm} via
\textsc{Matlab}'s numerical optimization features.
Since run-time is not a concern for an off-line design,
and in order to avoid local minima, we
run \texttt{fmincon} and \texttt{fminimax} to solve \eqref{eqn:cf_main_w} and \eqref{eqn:CRB_optm}, respectively, with 100 random initializations and pick
the solution with the smallest value of the cost function. In the following,  we refer to the design obtained by the SCF optimization approach from \eqref{eqn:cf_main_w} as Opt SCF, whereas the design obtained as a result of the CRB minimization from \eqref{eqn:CRB_optm} is referred to as Opt CRB. For the Opt SCF approach,
 we set $\ma{T} = \ma{A}^\herm \cdot \ma{A}$ as a target which is
the correlation function we would achieve with an $M$-element (uncompressed) UCA. For the Opt CRB approach, the threshold false detection probability is set to $0.05$ to be achieved at an SNR of $0$dB. 

\subsection{Performance analysis for a single source}

We begin by examining the performance of the optimized compressive arrays with respect to the attainable CRB and the sidelobe level in the case of a single source as discussed in Sections \ref{sec:SCF_design} and \ref{sec:design_crb}.

\begin{figure}[t!]
    \centering
    \includegraphics[width=0.9\linewidth]{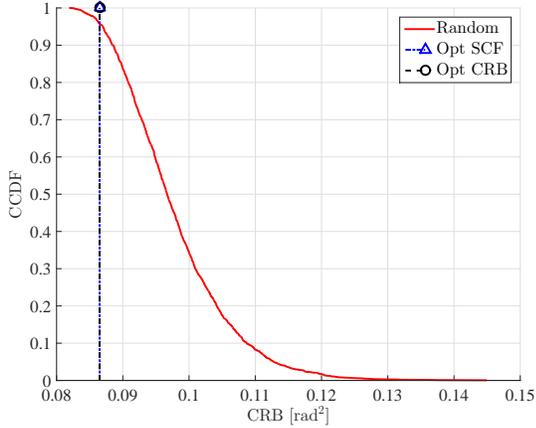}
    \caption{Comparison of the CRBs of the optimized compressive arrays versus the random ones. The CCDF of the CRB for 5000 random realizations are shown together with that of the optimized kernel.}
    \label{fig:crb_opt_rand}
    \vspace*{-0.2cm}
\end{figure}

\subsubsection{Comparison with random arrays}

	
Figure \ref{fig:crb_opt_rand} shows the achievable CRB 
of the compressive arrays with the optimized combining network and the random ones that have the same number of antennas and sampling channels while $\varphi_{m,n}$ are drawn uniformly at random from $(0, 2\pi]$. At a fixed SNR level of \unit[0]{dB}, an estimate of the Complementary Cumulative Distribution Function (CCDF) of the CRB obtained from $5000$ random realizations of $\bm \Phi $ is shown. The optimized networks for both approaches have been designed to achieve the same CRB. As evident from the
figure, the CRB of the optimized compressive arrays is almost in every case lower than that of the random ones. In other words, the random kernel can potentially provide a performance comparable to the optimized ones but with a very low probability. The CCDFs of the average sidelobe levels for the same scenario are shown in Figure \ref{fig:sl_opt_rand}. We observe that the random compressive arrays provide significantly higher sidelobe levels compared to both the SCF and the CRB-optimized ones. This supports the intuition that designing the combining matrix randomly results in sub-optimal performance.

Comparing the sidelobe level of the SCF and CRB-optimized compressive arrays, we can notice that the latter has a lower sidelobe level at a specific CRB. This is confirmed by the corresponding (analytic) probabilities of false detection depicted in Figure \ref{fig:Pd_SNR}.
For comparison Figure \ref{fig:Pd_SNR} also presents the results for the uncompressed UCAs with $N = 9$ and $N = 5$ (i.e., with smaller aperture size) as well as the average $P_{\rm d}$ for the random arrays.
As can be seen, both proposed approaches achieve lower probability of false detection $P_{\rm d}$ than the uncompressed UCA with a lower number of antennas, whereas the sparse arrays on average are significantly inferior to all the rest.  
It is worth nothing that for both proposed optimization approaches, the CRB and the sidelobe level (and hence the probability of false detection) can be controlled. 
In the SCF-based design this can be done by a proper choice of the weighting matrix in \eqref{eqn:cf_main_w}, whereas in the CRB-based approach the false detection probability is chosen directly.

\begin{figure}[t!]
    \centering
    \includegraphics[width=0.9\linewidth]{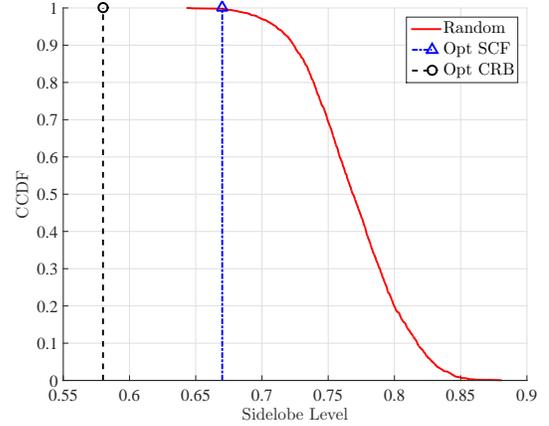}
    \caption{Comparison of the sidelobe levels of the optimized linear combining network versus the random ones. The CCDF of the mean sidelobe levels for 5000 random realizations are shown together with that of the optimized kernel.}
    \label{fig:sl_opt_rand}
    \vspace*{-0.2cm}
\end{figure}

\begin{figure}
    \centering
    \includegraphics[width=0.9\linewidth]{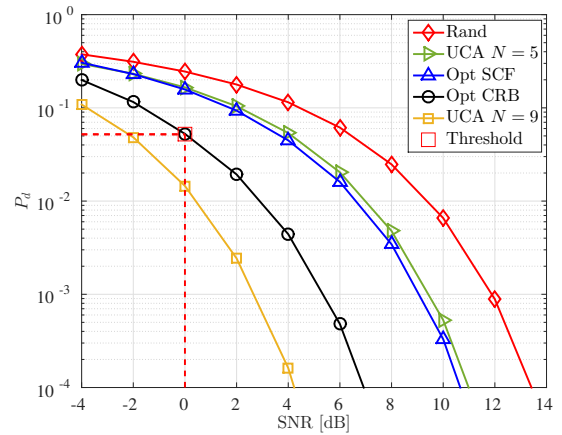}
    \caption{False detection probability of the uncompressed UCAs, compressive arrays, and the random ones.}
    \label{fig:Pd_SNR}
    \vspace*{-0.2cm}
\end{figure}

\begin{figure}[t]
    \centering
    \includegraphics[width=0.9\linewidth]{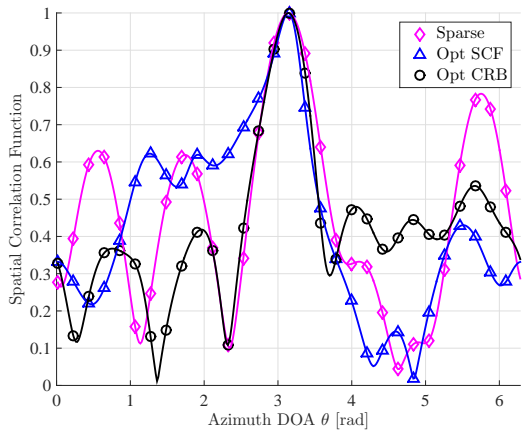}
    \caption{The spatial correlation function at a specific DOA of an optimized sparse array and a compressive array with the same number of receiver chains and an optimized combining network}
    \label{fig:sparse_vs_comp}
    \vspace*{-0.2cm}
\end{figure}
\begin{figure}[t]
    \centering
    \includegraphics[width=0.9\linewidth]{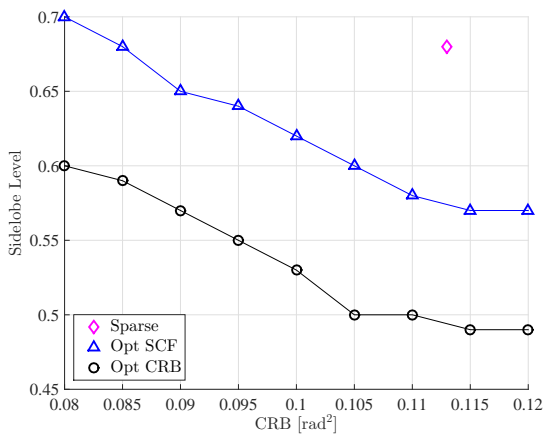}
    \caption{The CRB versus the sidelobe level for a compressive array with an optimized combining network and a sparse array (single diamond marker)}
    \label{fig:sl_crb}
    \vspace*{-0.2cm}
\end{figure}




\subsubsection{Comparison with sparse arrays}	
	Now we compare the compressive array to a sparse array that has the same number of receiver chains, i.e., for the considered scenario it means that for a sparse array $N=M=5$. According to \cite{GA:09}, we design the sparse array such that the positions of its elements are optimized towards obtaining a uniform sensitivity and desired CRB. 
	Figure \ref{fig:sparse_vs_comp} shows the spatial correlation functions at a specific DOA for the sparse array and an SCF-optimized compressive array (compression from $N = 9$ elements to $M = 5$ receiver chains) that achieves the same CRB (0.113) at the fixed SNR level of \unit[0]{dB}. It can be noted that the sidelobe level of the sparse array is relatively high compared to that of the compressive one (especially the Opt CRB design). Particularly, the mean level of the sidelobes for the optimized compressive array is 0.53 compared to 0.68 of the sparse array.
	
	Figure \ref{fig:sl_crb} shows the resulting trade-off between providing a good CRB and maintaining a low sidelobe level. One can see that the compressive array (with either of the optimization schemes) outperforms the sparse one and gives more degrees of freedom to tune the array design with respect to some desired properties (e.g., targeted CRB or sidelobe level). This figure also confirms once again that the Opt CRB approach outperforms the Opt SCF approach as it gives more control over the sidelobe level for a specific CRB.
	
	The superiority of the compressive arrays over the sparse ones with respect to adaptability is further highlighted in Figure \ref{fig:sl_crb_ad}. It presents the CRB and the sidelobe level of an optimized compressive array as a function of the number of antennas $N$. It is clear that the compressive arrays not only allow to control the CRB and the sidelobe level via an optimization of the combining network, but also by adding more antenna elements while the number of receiver channels is kept fixed. In Figure \ref{fig:sl_crb_ad}, we can see that the CRB can be improved significantly when the number of antenna elements is increased at the price of higher sidelobe levels. However the network can then be re-optimized for the new scenario (e.g., with the OPT CRB approach), aiming at a better suppression of the sidelobes while a certain level of CRB improvement is maintained.

\begin{figure}[t]
    \centering
    \includegraphics[width=0.9\linewidth]{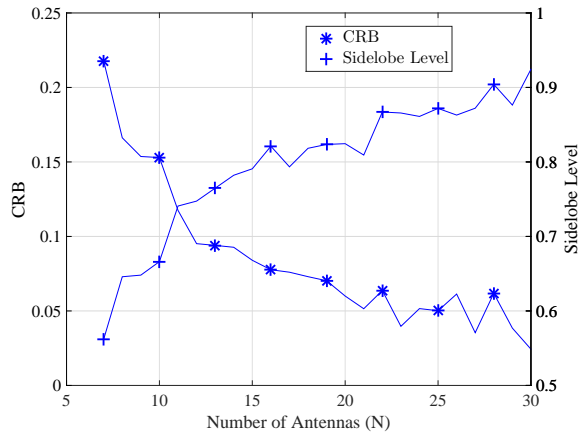}
    \caption{The CRB and the sidelobe level for a compressive array with different number of antenna elements and fixed number of channels}
    \label{fig:sl_crb_ad}
    \vspace*{-0.2cm}
\end{figure}

	\subsection{Performance analysis for multiple sources}

Now we examine the performance of the proposed optimized compressive array in the case of two sources impinging on the array from different DOAs. 
The power ratio between the two sources is set to $\alpha = {|s_2/s_1|^2} = -6$ dB, while their DOAs are $d$ radians apart, i.e., $\theta_1 = \theta$, $\theta_2 = \theta+d$ where $\theta$ scans the whole angular space. The two sources are inphase. 

Similar to Figure \ref{fig:crb_opt_rand}, Figure \ref{fig:ccdf_multi} shows the achievable CRB of the strongest path 
using the compressive arrays with the optimized combining network and the random ones that have the same number of antennas and sampling channels 
for an SNR level (with respect to the strongest source) of \unit[12]{dB}.  The CCDF of the CRB obtained from $5000$ random realizations of $\bm \Phi $ is shown for $d = 0.2$ and $d= 0.4$ radians. As discussed earlier, the Opt SCF and the Opt CRB designs both provide the same CRB and so only the Opt SCF is shown for clarity. The CCDF shows that the CRB of the optimized compressive arrays is again almost in every case lower than that of the random ones. As the sources get closer (e.g., $d=0.2$), the need for the optimized network increases as the probability to achieve acceptable properties (e.g., low CRB) by the random design gets lower.

\begin{figure}[t]
    \centering
    \includegraphics[width=0.9\linewidth]{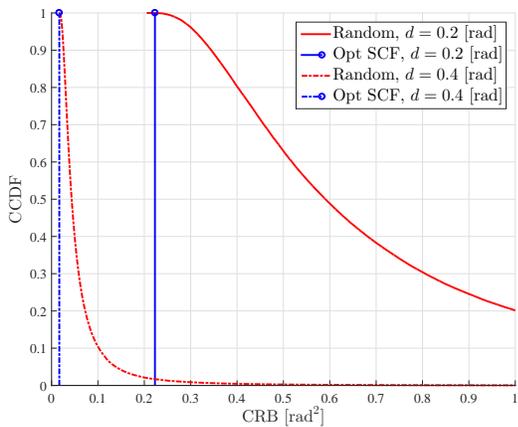}
    \caption{Comparison of the CRBs of the optimized compressive arrays versus the random ones for the case of two impinging sources. The CCDF of the CRB for 5000 random realizations are shown together with that of the optimized kernel.}
    \label{fig:ccdf_multi}
    \vspace*{-0.2cm}
\end{figure}


It has been proposed in Section \ref{sec:multi}, to extend the design based on CRB for the case of multiple sources. Considering the same set-up with two sources, and for a specific DOA of the first source, to do so the sidelobes in the correlation function have to be searched for all possible DOAs of the second source. This leads to a very high computational complexity. The same search strategy has to be performed for all possible DOAs of the first source. Therefore, for simplicity, we fix the second source DOA and perform the optimization similar to that of the single source case. Figure \ref{fig:Pd_SNR_Two_sources} shows the probabilities of false detection $P_{\rm d}$ for two sources. It can be seen that the design based on CRB shows superior performance in terms of lower false detection probability compared to that of the uncompressed $(N = 5)$-element UCA, compressive array with SCF based designed network, and the averaged random ones. Although the SCF based design can not be re-optimized for the multiple source case, it still provides a significantly lower probability of false detection compared to the random arrays and is comparable to a non-compressed \mbox{$(N=5)$-element UCA}.  

\begin{figure}[t!]
    \centering
    \includegraphics[width=0.9\linewidth]{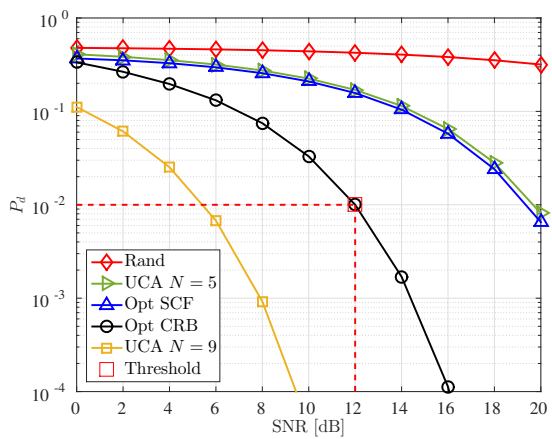}
    \caption{False detection probability of the uncompressed UCAs, compressive arrays, and the random ones with two source signals with  power ratio $\alpha = -6$dB.}
    \label{fig:Pd_SNR_Two_sources}
    \vspace*{-0.2cm}
\end{figure}

\section{Conclusions}\label{sec_concl}
In this paper we consider the design of compressive antenna arrays for direction of arrival (DOA) that aim to provide a larger aperture with a reduced hardware complexity compared to traditional array designs. We present an architecture of such a compressive array and introduce a generic system model that includes different options for the hardware implementation. 
We then focus on the choice of the coefficients in the analog combining network. Instead of choosing them randomly, as advocated by earlier work in this area, we propose a generic design approach for the analog combining network with the goal to obtain an array with certain desired properties, e.g., uniform sensitivity, low cross-correlation, or low variance in the DoA estimates. 
We exemplify the array design via two concrete examples. 
Our numerical simulations demonstrate the superiority of the proposed optimized compressive arrays to compressive arrays with randomly chosen combining kernels, as the latter result in very high sidelobes (which imply a higher probability of false detection) as well as higher CRBs. We also compare our optimized compressive array to a sparse array of the same complexity (i.e., same number of receiver channels $M$) and find that sparse arrays suffer from much higher sidelobes at the same CRB level. Also our proposed compressive array enjoys a high degree of adaptability since the combining weights can be altered to adjust the array to the current requirements, which is impossible for sparse arrays due to their static nature. 

\section*{Acknowledgment}
This work was partially supported by the Deutsche
Forschungsgemeinschaft (DFG) projects CLASS (grant MA 1184/23-1) and CoSMoS (grant GA 2062/2-1) and the Carl-Zeiss Foundation under the postdoctoral scholarship project  ``EMBiCoS''.

\appendices
\section{Proof of Theorem~\ref{thm:closedform} in Section \ref{sec:SCF_design}} \label{app:proof_closedform}

To prove the theorem we use the fact that for a unitary matrix $\ma{U}$ and an arbitrary square matrix $\ma{X}$
		we have $\left\|\ma{X} \cdot \ma{U} \right\|_{\rm F}
		= \left\|\ma{U} \cdot \ma{X} \right\|_{\rm F}
		= \left\|\ma{X} \right\|_{\rm F}.$
		Since $\ma{A}$ satisfies $\ma{A} \cdot \ma{A}^\herm = N \cdot \ma{I}_M$ 
		we can find a matrix $\ma{\bar{A}} \in \compl^{(M-N) \times N}$ such that
		$\ma{V} \stackrel{.}{=} 1/\sqrt{N} \cdot [\ma{A}^\trans, \ma{\bar{A}}^\trans]^\trans
		\in \compl^{N \times N}$ is a unitary matrix.
		Therefore, we have $\ma{V} \cdot \ma{A}^\herm 
		= \begin{bmatrix} \sqrt{N} \cdot \ma{I}_M, \ma{0}_{M \times {N-M}}\end{bmatrix}^\trans$.
		The cost function~\eqref{eqn:cf_main} can then be rewritten as
		\begin{align}
			& \left\|\ma{E}\right\|_{\rm F}^2 = 
			\left\|\ma{V} \cdot \ma{E}\cdot \ma{V}^\herm \right\|_{\rm F}^2 \notag \\
			& = \left\| 
			\begin{bmatrix} \sqrt{N} 
			\ma{I}_M \\ \ma{0}_{N-M \times M}\end{bmatrix}
			\ma{\Phi}^\herm 
			\ma{\Phi}
			\begin{bmatrix} \sqrt{N} 
			\ma{I}_M,  \ma{0}_{M \times N-M}\end{bmatrix}
			- 
			\ma{V} 
			\ma{T} 
			\ma{V}^\herm
			\right\|_{\rm F}^2
			\notag \\
			& = 
			\left\| 
			\begin{bmatrix} 
			N \ma{\Phi}^\herm \cdot \ma{\Phi} & \ma{0}_{M \times N-M} \\
			\ma{0}_{N-M \times M} & \ma{0}_{N-M \times N-M}
			\end{bmatrix}
			-
			N \cdot
			\begin{bmatrix}  \ma{A} \\ \ma{\bar{A}} \end{bmatrix}
			\ma{T} 
			\begin{bmatrix}  \ma{A}^\herm, \ma{\bar{A}}^\herm \end{bmatrix}
			\right\|_{\rm F}^2
			\notag \\
			& = 
			\left\| 
			N \cdot
			\begin{bmatrix} 
			\ma{\Phi}^\herm \cdot \ma{\Phi} - \ma{A} \cdot \ma{T} \cdot\ma{A}^\herm & -\ma{A} \cdot\ma{T} \cdot\ma{\bar{A}}^\herm \\
			- \ma{\bar{A}}\cdot \ma{T} \cdot\ma{A}^\herm & - \ma{\bar{A}} \cdot\ma{T} \cdot\ma{\bar{A}}^\herm
			\end{bmatrix}
			\right\|_{\rm F}^2
			\notag \\
			& =
			N^2 \cdot
			\left\|
			\ma{\Phi}^\herm \cdot \ma{\Phi}
			-
			\ma{S}
			\right\|_{\rm F}^2
			+ {\rm const}, \label{eqn:proof_ir1}
		\end{align}
using the short-hand notation $\ma{S} = \ma{A} \cdot \ma{T} \cdot\ma{A}^\herm$.
Equation~\eqref{eqn:proof_ir1} demonstrates that the optimization problem is equivalent
to finding the best approximation of the matrix $\ma{S}$ by the matrix $\ma{\Phi}^\herm \cdot \ma{\Phi}$.
Since $\ma{\Phi}$ is an $m \times M$ matrix, the rank of the $M \times M$ matrix 
$\ma{\Phi}^\herm \cdot \ma{\Phi}$ is less than or equal to $m<M$. Therefore, \eqref{eqn:proof_ir1}
represents a low-rank approximation problem. According to the Eckart-Young theorem, its 
optimal solution is given by truncating the $M-m$ smallest
eigenvalues of $\ma{S}$.

\section{ Proof of Corollary \ref{cor:roworth} in Section \ref{sec:SCF_design}}
\label{app:proof_corollary}
 The sampled version of \eqref{eqn:manifold_ideal} is given by a scaled identity matrix,
		i.e., $\ma{T} = C\cdot \ma{I}_N$. Since $\ma{A}$ is row-orthogonal it follows that
		$\ma{S} = \bm{A} \cdot \bm{T} \cdot \bm{A}^\herm = C \cdot N \cdot \ma{I}_M$.
		As all eigenvalues of $\ma{S}$ are equal to $C \cdot N$, its eigenvalue decomposition
		can be written as $\ma{S} = \ma{U} \cdot (C \cdot N\cdot \ma{I}_M) \cdot \ma{U}^\herm$,
		where $\ma{U} \in \compl^{M \times M}$ is an arbitrary unitary matrix. Truncating the $M-m$
		``smallest'' eigenvalues, we obtain $\ma{S}_m = C \cdot N \cdot \ma{U}_m \cdot \ma{U}_m^\herm$,
		where $\ma{U}_m \in \compl^{M \times m}$ contains the first $m$ columns of $\ma{U}$.
		Invoking Theorem~\ref{thm:closedform}, we have
		 $\bm{\Phi}_{\rm opt}^\herm \bm{\Phi}_{\rm opt} = C \cdot N \cdot \ma{U}_m \cdot \ma{U}_m^\herm$
		and therefore $\bm{\Phi}_{\rm opt}$ is a scaled version of $\ma{U}_m^\herm$,
		which proves the claim.

\section{Proof of Theorem \ref{thr:prblt} in Section \ref{sec:design_crb}}
\label{app:prooffalsedetecion}

	Denote by $X_q(t) = D(\theta_{0})-D(\theta_{q})$, where $D(\theta)  = \left| \bm{\tilde{a}}^{\rm H}\bm{\tilde{y}}(t) \right |^2 $ and $D(\theta_q)  = \left| \bm{\tilde{a}}_q^{\rm H}\bm{\tilde{y}}(t) \right |^2 $. Then, $X_q(t)$ is a random variable that 
	we can write as
	\begin{align}
	X_q(t)&  =  \left| \bm{\tilde{a}}_0^{\rm H}\bm{\tilde{y}}(t) \right |^2 - \left| \bm{\tilde{a}}_q^{\rm H}\bm{\tilde{y}}(t) \right |^2 \nonumber \\
	& = \bm{\tilde{y}}^{\rm H}(\bm{\tilde{a}}_0\bm{\tilde{a}}_0^{\rm H} - \bm{\tilde{a}}_q\bm{\tilde{a}}_q^{\rm H})\bm{\tilde{y}} = \bm{\tilde{y}}^{\rm H}\bm{D}\bm{\tilde{y}} 
	\label{eqn:qdrc_form}
	\end{align}
	where  $\bm{\tilde{y}} = \bm{\tilde{a}}_0{s}(t) + \tilde{\bm{n}}(t) $ 
 is a complex random vector with  a non-zero mean and a covariance matrix $\bm{R}_{\rm nn}$.  Assuming that $\bm{v}(t)$ and $\bm{w}(t)$ are white with elements that have variance $\sigma_1^2$ and $\sigma_2^2$,
 $\bm{\tilde{y}}$ becomes  a complex-Gaussian random vector with mean equal to $\bm{\tilde{a}}_{\rm 0} s_{\rm 0}$ and covariance $\bm{R}_{\rm nn} = \sigma_1^2  \bm{\Phi} \bm{\Phi}^\herm + {\sigma_2^2}  \bm{I}$. 
This said, $X_{q}(t)$ is  a chi-square random variable that,
	due to the structure\footnote{Note that since the matrix $\bm{D}$ in \eqref{eqn:qdrc_form} results from the subtraction of the outer products of two vectors $\bm{\tilde{a}}_{\rm 0}$ and $\bm{\tilde{a}}_{\rm q}$, it is a complex symmetric matrix that whose maximum rank is $2$ 
	.} of $\bm{D}$ and the fact that $\tilde{\bm{y}}$ is non-zero mean, 
	has a so-called non-central indefinite quadratic form \cite{MP:92}.

	In order to compute the probability $ \textup{Prob}(X_{q}(t) < 0)$, we need to derive the distribution of $X_{q}(t)$. Since the covariance matrix of $\tilde{\bm{y}}$  is colored, it is convenient to express $X_q(t)$ as 
	\begin{equation}
	    	X_q(t)  = \bm{\tilde{y}}_{\rm w}^{\rm H} \bm{B} \bm{\tilde{y}}_{\rm w},
	    	\label{eqn:qdr_form_diag}
	\end{equation}
	where $\bm{\tilde{y}}_{\rm w} = \bm{R}_{\rm nn}^{-1/2}\bm{\tilde{y}}$ contains
	pre-whitened observations whose covariance is an identity matrix and $\bm{B} = \bm{R}_{\rm nn}^{1/2}\bm{D}\bm{R}_{\rm nn}^{1/2}$.
This, way the quadratic form  \eqref{eqn:qdrc_form} is reduced to a diagonal form in independent random variables with unit variance \cite{FA:02}. 
%

	By representing  $\bm{B}$ via its eigen value decomposition (EVD) as $\bm{B} = \bm{U}_{\rm n}\bm{\lambda}\bm{U}_{\rm n}^{\rm H}$ with  $\bm{U}_{\rm n}$ and $\bm{\lambda}$ being the unitary matrix and the diagonal matrix consisting of eigenvalues, respectively, then we can compute the moment generating function (MGF) shown in \cite{GLT:60}, \cite{DR:96} as
	\begin{align}
	 \Psi_{X_{q}}(s) = \frac{ \textup{exp}\Big (\sum_{r=1}^{R} \frac{\mu_{r}^2 \lambda_{r} s}{1-\lambda_{r}s} \Big)}{ \prod_{r=1}^{R} (1- \lambda_{r} s)}.
	 \label{eqn_mfg}
	\end{align}
	Here, $\lambda_r$ is the $r$th eigenvalue of $\bm{B}$,  
	$R$ denotes the rank $\bm{D}$, and $\mu_r$ is the $r$th element of the vector of non-centrality parameters $\bm{\mu}$ defined as
	\begin{align}
	\bm{\mu} = \bm{U}_{\rm n}^{\rm H}\bm{R}_{\rm n}^{-1/2}\bm{r},
	\label{eqn:mu}
	\end{align}
	where $\bm{r}=\bm{\tilde{a}}_{0} s_{0}$.

	The PDF of $X_{q}$ can now be obtained  by computing the inverse Laplace transform of  $ \Psi_{X_{\rm q}}(s)$, whereas the probability of false detection is then obtained by integrating the resulting PDF. %
    In order to compute the integral we apply an iterative approach for numerical integration from \cite{MLP:02} 
    that utilizes the saddle point technique from \cite{CH:86}. 
	In this technique, the integration path is chosen such that is passes through the saddle point of the integrand on the real axis \cite{FA:02}. Since the integrand is convex, a single  saddle point $s =  s_{\rm p}$  exists in $\realof{s}> 0$. Furthermore, it can be easily computed by Newton search method as $s_{\rm p} \leftarrow s_{\rm p} - (\Psi'(s_{\rm p})/\Psi''(s_{\rm p}))$, where $\Psi'(s_{\rm p})$, $\Psi''(s_{\rm p})$ are the first and second order derivatives of  $\Psi(s) = \ln(\Psi_{X_{\rm q}}(s)/s)$ evaluated at the saddle point $s_{\rm p}$ \cite{MLP:02}.

Therefore, using the Gauss-Chebyshev quadrature \cite{MLP:02}, the probability of false detection can be obtained as 
	\begin{align}
	P_{q} \approx \frac{1}{2G} \sum_{g=1}^{G} \hat{\Psi}\Bigg( \frac{(2g-1)\pi}{2G}\Bigg)
	\end{align}
	where $\hat{\Psi}(\tau ) = (1-j \tan(\tau /2))\Psi_{X_q}\big(-s_{\rm p}(1+j \tan(\tau /2))\big)$ and $G$ is the number of steps that determines the accuracy of the integration.

\section{CRB for compressed array}
\label{app:CRLB}

In this section, we present results of the CRB derived in \cite{HLVT:02} for the receiver model shown in Fig. \ref{fig:SystemModel} where the noise vectors $\bm{v}(t)$ and $\bm{w}(t)$ are assumed to be white with covariances $\bm{R}_{\rm vv} = \sigma_1^2 \mathbf{I}_N$ and $\bm{R}_{\rm ww} = \sigma_2^2 \mathbf{I}_M$, respectively. The associated CRB matrix is then found to be
\begin{equation} 
\textup{CRB}(\bm{\Phi}, \bm{\theta}) = \sigma_1^2\big(2\realof{\bm{F} \odot \bm{R}_s^{\rm T}})^{-1},
\label{eqn:CRLB}
\end{equation} 
where $\odot$ denotes Shur (element-wise) matrix product, $\bm{R}_s =  \bm{s}(t)\bm{s}^{\rm H}(t)$ is the signal covariance matrix, and $\bm{F}$ is a matrix that depends on the array beampattern and the combining matrix $\bm{\Phi}$ as 
\begin{equation} \bm{F} = \bm{D}^{\rm H}\bm{\Phi}^{\rm H}\bm{Z}\bm{\Phi}\bm{D}.
\label{eqn:CRB_Fmtx}
\end{equation}
In \eqref{eqn:CRB_Fmtx}, $\bm{Z} = \bm{Q}\Big( \mathbf{I}_N- \bm{\tilde{A}}\Big(\bm{\tilde{A}}^{\rm H}\bm{Q}\bm{\tilde{A}}  \Big)^{-1}\bm{\tilde{A}}^{\rm H} \bm{Q}\Big)$,
\mbox{$\bm{D} = \big[ {\partial \bm{a}(\theta_{0})}/{\partial \theta_{0}}, {\partial \bm{a}(\theta_{1})}/{\partial \theta_{1}},...,{\partial \bm{a}(\theta_{K-1})}/{\partial \theta_{K-1}} \big]$, and}  $\bm{Q} = (\bm{\Phi}\bm{\Phi}^{\rm H} + \beta \mathbf{I}_N)^{-1}$ where $\beta = \frac{\sigma_2^2}{\sigma_1^2}$. Since we consider only a single source for solving the optimization problem in \eqref{eqn:CRB_optm}, $\bm{F}$ reduces to a scalar $F$ and $\bm{R}_s$ to $R_{ss} = \|s(t)\|^2$. Therefore, for a single source \eqref{eqn:CRLB} becomes
\begin{align} \text{CRB}(\bm{\Phi},\theta_{0}) = \sigma_1^2\Big(2F R_{ss} \Big )^{-1}  = \frac{1}{2F\rho},
 \label{single_doa}\end{align}
where $\rho = \frac{R_{ss}}{\sigma_1^2}$ is the input SNR.


\smallskip
\bibliographystyle{unsrt}
\bibliography{refs}

\end{document}